\documentclass[11pt]{article}
\usepackage[margin=1in]{geometry}
\usepackage{amsmath,amssymb,amsthm,mathtools}
\usepackage{booktabs}
\usepackage{enumitem}
\usepackage{microtype}
\usepackage{array}
\usepackage{xcolor}
\usepackage[numbers]{natbib}
\usepackage{thmtools}

\makeatletter
\@ifundefined{newcounteralias}{}{%
  \renewcommand\thmt@autorefsetup{\@xa\def\csname\thmt@envname autorefname\@xa\endcsname\@xa{\thmt@thmname}}%
}
\makeatother

\usepackage{physics}
\usepackage{authblk}
\usepackage{hyperref}
\usepackage{cleveref}

\declaretheorem[name=Theorem, numberwithin=section]{theorem}

\declaretheorem[name=Lemma, sibling=theorem]{lemma}
\declaretheorem[name=Proposition, sibling=theorem]{proposition}
\declaretheorem[name=Corollary, sibling=theorem]{corollary}

\declaretheorem[name=Definition, sibling=theorem, style=definition]{definition}
\declaretheorem[name=Problem, sibling=theorem, style=definition]{problem}

\declaretheorem[name=Remark, sibling=theorem, style=remark]{remark}

\newcommand{\F}{\mathbb F}
\newcommand{\wt}{\operatorname{wt}}
\newcommand{\Aut}{\operatorname{Aut}}
\newcommand{\PAut}{\operatorname{PAut}}
\newcommand{\GL}{\operatorname{GL}}
\newcommand{\GLtwo}[1]{\GL(#1,\mathbb F_2)}
\newcommand{\Span}{\operatorname{span}}

\newcommand{\supp}{\operatorname{supp}}

\newcommand{\SURJ}{\mathrm{SURJ\text{-}QA}}
\newcommand{\cP}{\mathcal P}

\title{Phantom Codes: Hardness, Rate Optimal qLDPC Constructions, and Distance Limits}
\author[1]{Rui Mao}
\author[1]{Weixiao Sun}
\author[1]{Shengyu Zhang \thanks{shengyzhang@tencent.com}}
\affil[1]{Tencent Quantum Laboratory}

\date{\today}

\begin{document}
\maketitle

\begin{abstract}
  An $[[n,k,d]]$ stabilizer code is \emph{phantom} if every in-block logical CNOT gate can be realized by a permutation of its physical qubits.
  This eliminates the large and complicated physical overhead normally required for logical entangling gates.
  Yet this symmetry is highly restrictive: phantom codes are rare, the number of logical qubits is limited to $k=O(\log n)$, and no phantom qLDPC family with growing logical dimension was previously known.

  We make three contributions.
  \begin{enumerate}
    \item
          We prove that recognizing phantomness of a given stabilizer code is at least as hard as Graph Isomorphism, even for CSS codes encoding only $k=2$ logical qubits.
    \item
          We propose the first rate optimal phantom qLDPC families: for every fixed $D$, our CSS families achieve the maximal logical scaling $k=\Theta(\log n)$ and distance $d\geq D$.
    \item
          We show a distance no-go theorem: every phantom code family with $k=\omega(\sqrt{\log n})$ and check weight $w=O(1)$ satisfies $d\leq w$ for sufficiently large $n$.
  \end{enumerate}
  Thus our fixed-distance qLDPC families are distance-scaling optimal at the maximal logical scaling $k=\Theta(\log n)$.
\end{abstract}

\tableofcontents

\section{Introduction}\label{sec:intro}

Fault-tolerant logical entangling gates are a major source of overhead in quantum computation \cite{koh2026entangling}.
Code symmetries offer a way to reduce this overhead: a permutation of physical qubits that preserves the code can implement a nontrivial logical operation.
Such an operation requires only a change of qubit labels, with no additional quantum operation.
Sayginel et al. developed systematic methods for extracting logical Clifford gates from code automorphisms, while Guyot and Jaques studied the addressability of logical operations through physical permutations \cite{sayginel2025fault,guyot2026addressability}.
Koh et al. took this idea further by introducing \emph{phantom codes}, in which every logical CNOT gate within a single code block can be implemented by relabelling physical qubits \cite{koh2026entangling}.

However, this convenience of phantom codes comes with restrictions.
First, phantom codes are rare, and their numerical discovery is computationally demanding.
Writing $n$ for the number of physical qubits and $k$ for the number of logical qubits.
Koh et al. exhaustively enumerated approximately $2.71\times10^{10}$ inequivalent CSS codes with $n\leq14$, finding only $132\,305$ phantom codes with distance at least two \cite[App.~C]{koh2026entangling}.
Reaching this modest block length required approximately $10\,000$ computational cores running for $1.5$ months \cite[App.~C.4]{koh2026entangling}.
These limitations motivate the question of how difficult it is to recognize phantomness in a given code, whose computational complexity remained open.

Second, phantomness severely limits the number of logical qubits that can be encoded.
Morris and Malz proved that binary phantom codes of distance at least two satisfy $k\leq\log_2(n+1)$ when $k\neq4$, and classified the nontrivial CSS codes saturating this bound \cite{morris2026constraints}.
Thus the largest possible asymptotic logical scaling is $k=\Theta(\log n)$.
The next natural question is whether phantom codes at this maximal logical scaling can also achieve favorable error-correction parameters, particularly sparse checks and large distance.

Third, no phantom quantum low-density parity-check (qLDPC) family with growing logical dimension was previously known \cite{koh2026entangling,he2026logical,morris2026constraints}.
A qLDPC family keeps both check weights (the number of qubits involved in each stabilizer check) and qubit degrees (the number of checks involving each qubit) bounded independently of the code length.
The simplex--repetition HGP construction does give qLDPC families of growing code length when $k$ is fixed, but its check weights and qubit degrees increase with $k$ \cite[App.~B]{he2026logical}.
It remained unclear whether a phantom qLDPC family with growing logical dimension exists and, if so, whether it can achieve maximal logical scaling, large distance, or both.

In this paper, we address three questions arising from these restrictions.
\emph{How hard is it to recognize phantomness?
  Do growing phantom qLDPC families exist at maximal logical scaling?
  If so, what distance can they achieve?}
Our results can be summarized as follows.

\paragraph{Main results.}
First, we show that deciding whether a code is phantom is at least as hard as the Graph Isomorphism (GI) problem.

\begin{restatable}[GI-hardness of recognizing phantomness]{theorem}{gihardness}
  \label{thm:GI-hard}
  Graph Isomorphism admits a polynomial-time many-one reduction to recognizing binary phantom stabilizer codes, even when the input is CSS and the logical dimension is fixed to $k=2$.
\end{restatable}

Second, we construct phantom qLDPC families with any prescribed constant distance $d$ and the maximal logical scaling $k=\Theta(\log n)$.
Here distance $d$ is the smallest number of qubits supporting a nontrivial logical Pauli operation and measures the code's protection against errors.

\begin{restatable}[Phantom qLDPC codes at arbitrary fixed distance]{theorem}{phantomqldpc}
  \label{thm:qldpc-fixedD}
  For every fixed integer $D\geq1$, there is a family of binary phantom qLDPC CSS codes satisfying $k=\Theta(\log n)$ and $d\geq D$.
\end{restatable}

Third, we show that these families are optimal in distance scaling: in the regime below, the distance is eventually bounded by the check weight.
\begin{restatable}[Distance no-go]{theorem}{stabilizernogo}
  \label{thm:stabilizer-no-go}
  Let $\{S_t\}$ be an unbounded family of binary phantom stabilizer codes with lengths $n_t$, logical dimensions $k_t=\omega(\sqrt{\log n_t})$, generator weights $w(S_t)=O(1)$, and distances $d(S_t)$.
  Then $d(S_t)\leq w(S_t)$ for all sufficiently large $t$.
\end{restatable}

The second theorem constructs phantom qLDPC families with any prescribed constant distance at the maximal logical scaling $k=\Theta(\log n)$.
The third theorem applies already when $k=\omega(\sqrt{\log n})$ and implies that distance cannot grow without check weight also growing.
In particular, it makes the construction optimal in distance scaling.

\paragraph{Relation to prior work}
On the complexity side, our reduction is closely related in spirit to the classical result of Petrank and Roth that binary code equivalence is GI-hard \cite{petrank1997code}.
However, the decision problem in the quantum case is different: rather than asking whether two codes are permutation-equivalent, we ask whether the physical permutation group of a single stabilizer code realizes the full standard logical CNOT group.
The reduction shows that this global logical-action question remains GI-hard even when the code is CSS and the logical dimension is two.

Our qLDPC construction combines various ingredients in a phantom-compatible way.
Motivated by earlier phantom constructions \cite{koh2026entangling,he2026logical}, we realize the dual Hamming stabilizer space using the variable-copying degree-reduction gadget of \cite{hastings2021quantum,hastings2021fiber}, and then use fixed-block concatenation \cite{knill1996concatenated,koh2026entangling} to raise both sector distances.
To our knowledge, this gives the first growing-$k$ phantom family with both bounded check weight and bounded qubit degree, and the first such family at every prescribed constant distance.

\paragraph{Organization.}
\Cref{sec:setup} gives the formal definitions, notation, and elementary facts.
\Cref{sec:GI} proves GI-hardness.
\Cref{sec:constructions} constructs explicit phantom qLDPC families with maximal logical scaling and arbitrary constant distance.
\Cref{sec:obstructions} proves the distance no-go theorem first for CSS codes and then for general stabilizer codes.
\Cref{sec:conclusions} summarizes the implications of these results and discusses open problems.

\section{Preliminaries}\label{sec:setup}

Throughout, all vector spaces are over $\F_2$.
For a binary vector $v$, its support is $\supp(v):=\{i:v_i\neq0\}$, and its Hamming weight is $\wt(v):=|\supp(v)|$.
The standard bilinear form is \[ \langle x,y\rangle = \sum_{i=1}^n x_i y_i\pmod 2. \]
For a subspace $C\leq\F_2^n$, its orthogonal complement is
\[
  C^\perp=\{x\in\F_2^n:\langle x,c\rangle=0\text{ for all }c\in C\}.
\]
A coordinate permutation $\pi\in S_n$ acts on $\F_2^n$ by permuting coordinates.
Such permutations preserve Hamming weight and the bilinear form.

For a finite set $I$, write
\begin{equation}
  \label{eq:pauli-space}
  \cP(I):=\F_2^I\oplus\F_2^I
\end{equation}
for the binary Pauli space modulo phases, with symplectic form
\[
  \langle(x,z),(x',z')\rangle_s :=\langle x,z'\rangle+ \langle z,x'\rangle.
\]
For $p=(x,z)$, define its qubit support and weight by
\[
  \supp_q(p):=\supp(x)\cup\supp(z), \qquad \wt_q(p):=|\supp_q(p)|.
\]
For a subspace $S\leq\cP([n])$, its symplectic orthogonal complement is
\[
  S^{\perp_s}:=\{p\in\cP([n]):\langle p,s\rangle_s=0\text{ for every }s\in S\}.
\]

\begin{definition}[Phantom stabilizer code]
  \label{def:phantom-stabilizer}
  A phase-free binary stabilizer code is specified by an isotropic subspace $S\leq\cP([n])$.\footnote{For a bilinear form $B$ on a space $W$, a subspace $U$ is isotropic if $B(U,U)=0$ and Lagrangian if maximal isotropic.}
  It encodes $k=n-\dim S$ qubits and has the nondegenerate logical Pauli space
  \[
    L:=S^{\perp_s}/S.
  \]
  A physical coordinate permutation acts diagonally on the $X$ and $Z$ coordinates of $\cP([n])$: $\pi((x,z)) = (\pi(x), \pi(z))$.
  The permutation automorphism group of $S$ is \[ \PAut(S):=\{\pi\in S_n:\pi(S)=S\}. \]
  Fix a $k$-dimensional space $V$ and a symplectic identification $L\cong V\oplus V^*$, with restricted symplectic form
  \begin{equation}
    \label{eq:canonical-logical-symplectic-form}
    \omega((x,f),(y,g)) = g(x) + f(y).
  \end{equation}
  The standard logical CNOT subgroup is
  \[
    \rho(\GL(V))\leq\operatorname{Sp}(V\oplus V^*), \qquad \rho(F)(x,f):=(Fx,f\circ F^{-1}).
  \]
  We call $S$ a \emph{phantom stabilizer code}, or simply \emph{phantom}, if the image of $\PAut(S)$ on $L$ contains $\rho(\GL(V))$ for some such identification.
\end{definition}

When $S$ is CSS, \Cref{def:phantom-stabilizer} reduces to surjectivity on the logical $X$ quotient introduced below, and the action on the logical $Z$ quotient is contragredient.

\begin{definition}[Phantom CSS code]
  \label{def:phantom-css}
  A phase-free binary CSS code is specified by its $X$- and $Z$-type stabilizer spaces $V_X,V_Z\leq\F_2^n$ satisfying $V_X\subseteq V_Z^\perp$.
  It encodes $k=n-\dim V_X-\dim V_Z$ qubits and has logical quotients \[ Q_X:=V_Z^\perp/V_X, \qquad Q_Z:=V_X^\perp/V_Z. \]
  The permutation automorphism group of the CSS code is \[ \PAut(V_X,V_Z):=\{\pi\in S_n:\pi(V_X)=V_X,\ \pi(V_Z)=V_Z\}. \]
  Every $\pi\in \PAut(V_X,V_Z)$ induces an automorphism of $Q_X$, giving a homomorphism \[ A:\PAut(V_X,V_Z)\longrightarrow \GL(Q_X)\cong \GLtwo{k}. \]
  We call $A$ the \emph{induced quotient action}; and $(V_X,V_Z)$ a \emph{phantom CSS code}, or simply \emph{phantom}, if $A$ is surjective.
\end{definition}

The next lemma shows that surjectivity on $Q_X$ automatically gives surjectivity on $Q_Z$.

\begin{lemma}
  If $A$ is surjective on $Q_X$, then the induced action on $Q_Z$ is also surjective.
\end{lemma}

\begin{proof}
  The standard bilinear form induces a perfect pairing \[ Q_X\times Q_Z\longrightarrow\F_2, \qquad (x+V_X,z+V_Z)\mapsto \langle x,z\rangle. \]
  Indeed, the annihilator of $V_X$ is $V_X^\perp$, and quotienting by $V_Z$ and $V_X$ removes precisely the radicals.\footnote{For a bilinear pairing $B:E\times F\to\mathbb F$, the left radical is $\{e:B(e,F)=0\}$. The right radical is similarly defined.}
  Thus $Q_Z$ identifies with the dual vector space $Q_X^*$.
  A coordinate permutation preserves the pairing, so the action on $Q_Z$ is the contragredient of the action on $Q_X$.
  The contragredient map $g\mapsto g^{-\top}$ is an automorphism of $\GLtwo{k}$.\footnote{Here $g^{-\top}:=(g^{-1})^\top=(g^\top)^{-1}$ denotes the inverse transpose of $g$.}
  Therefore surjectivity on $Q_X$ implies surjectivity on $Q_Z$.
\end{proof}

\paragraph{Check sparsity.}
For a stabilizer generating basis $B\subseteq S$, sparsity is measured by the maximum row and column weights of its support-incidence matrix.
Explicitly, define \[ w(B):=\max_{b\in B}|\supp_q(b)|, \qquad \Delta(B):=\max_{i\in[n]}\bigl|\{b\in B:i\in\supp_q(b)\}\bigr|, \]
where $\supp_q$ denotes the qubit support defined above.
Thus $w(B)$ is the maximum number of qubits in a stabilizer check, while $\Delta(B)$ is the maximum number of checks incident on a qubit.
A family is qLDPC if it admits stabilizer generating sets with $w=O(1)$ and $\Delta=O(1)$.
We also use the intrinsic check weight \[ w(S):=\min_{B: \text{ basis of } S}w(B). \]
For a CSS pair $(V_X,V_Z)$ we similarly set \[ w(V_X,V_Z):=\max\{w(V_X),w(V_Z)\}, \qquad w(V):=\min_{B: \text{ basis of } V}\max_{b\in B}|\supp(b)|. \]

\paragraph{Code Distance.}
For a general stabilizer code, the distance is \[ d(S):=\min\{\wt_q(p):p\in S^{\perp_s}\setminus S\}. \]
For CSS codes, this becomes the minimum of the two sector distances
\[
  d_X:=\min\{\wt(x):x\in V_Z^\perp\setminus V_X\}, \qquad d_Z:=\min\{\wt(z):z\in V_X^\perp\setminus V_Z\}, \qquad d:=\min\{d_X,d_Z\}.
\]

\paragraph{Phase convention}
We suppress scalar Pauli phases in \Cref{def:phantom-stabilizer,def:phantom-css} because, after fixing a reference signed lift $\widehat S_0$ of a binary stabilizer space $S$ (the canonical all-$+$ lift in the CSS case) every other signed lift $\widehat S$ satisfies \[ \widehat S_0=P\widehat S P^\dagger \] for some physical Pauli string $P$ \cite{gottesman1997stabilizer,dehaene2003clifford}.
Thus the phase-free binary description entails no loss for the code parameters or logical actions considered below.

Phases matter only if one requires the bare coordinate permutation to preserve the signed stabilizer group exactly.
In such case, $\PAut(\widehat S)$ can be a strict subgroup of $\PAut(S)$.
All results below concerning $n$, $k$, distance, generator weight, qLDPC sparsity, constructions, GI-hardness, and no-go bounds remain valid under the Pauli-frame convention.
Under the stricter bare-permutation convention, the constructions remain valid for the canonical all-$+$ lift, while the no-go results continue to hold because every exact phase-preserving symmetry is also an element of the corresponding binary permutation automorphism group.

\section{GI-hardness of recognizing phantom stabilizer codes}\label{sec:GI}

It is enough to prove hardness on CSS inputs, for which the phantom condition of \Cref{def:phantom-css} is the following quotient-surjectivity problem.

\begin{problem}[Surjective quotient action]
  An instance consists of generator matrices for binary linear codes $V_X,V_Z\leq\F_2^n$ satisfying $V_X\subseteq V_Z^\perp$.
  Decide whether \[ A(\PAut(V_X,V_Z))=\GLtwo{k}, \qquad k=n-\dim V_X-\dim V_Z. \]
  Let $\SURJ_k$ denote the restriction to instances with quotient dimension $k$.
\end{problem}

The main result of this section is the following.

\gihardness*

The proof of \Cref{thm:GI-hard} breaks into \Cref{sec:connectivity-gadget,sec:graph-code,sec:GI-reduction-map}.
In \Cref{sec:logical-relabelling}, we prove a similar theorem showing that GI-hardness persists even for the weaker requirement that physical permutations realize every logical SWAP.

\subsection{A connectivity gadget}\label{sec:connectivity-gadget}

For a graph $X$, write $X\vee K_2$ for the join of $X$ with a two-vertex clique: add two new adjacent vertices and connect each of them to every vertex of $X$.

\begin{lemma}
  \label{lem:join}
  For non-empty simple graphs $X,Y$, \[ X\cong Y\quad\Longleftrightarrow\quad X\vee K_2\cong Y\vee K_2. \]
  Moreover, $X\vee K_2$ is connected and has minimum degree at least $2$.
\end{lemma}

\begin{proof}
  The forward implication is immediate.
  For the converse, observe \[ \overline{X\vee K_2}=\bar X\sqcup 2K_1. \]
  Thus an isomorphism $X\vee K_2\cong Y\vee K_2$ yields
  \[
    \bar X\sqcup 2K_1\cong \bar Y\sqcup 2K_1.
  \]
  By uniqueness of the decomposition of a finite graph into connected components, the multisets of nontrivial connected components of $\bar X$ and $\bar Y$ agree, and the number of isolated vertices in $\bar X$ and $\bar Y$ also agrees because both sides have had exactly two isolated vertices added.
  Hence $\bar X\cong\bar Y$, and therefore $X\cong Y$.

  The two new vertices are adjacent to one another and to all old vertices.
  Every old vertex is adjacent to both new vertices, so every vertex has degree at least $2$, and the graph is connected.
\end{proof}

Given a GI instance $(X,Y)$, if the graphs have different numbers of vertices we may map it to a fixed no-instance.
Hence assume they have the same number of vertices and set \[ X'=X\vee K_2, \qquad Y'=Y\vee K_2. \]
Now form
\[
  Z=X'\sqcup Y'\sqcup Y'.
\]
Let the three connected components be $H_1=X'$, $H_2=Y'$, and $H_3=Y'$.

\begin{lemma}
  \label{lem:component-group}
  Let $P(Z)\leq S_3$ be the permutation group induced by $\Aut(Z)$ on the set of connected components $\{H_1,H_2,H_3\}$.
  Then
  \[
    P(Z)=
    \begin{cases}
      S_3,                         & X\cong Y,     \\
      \langle(23)\rangle\cong C_2, & X\not\cong Y.
    \end{cases}
  \]
\end{lemma}

\begin{proof}
  If $X\cong Y$, then $X'\cong Y'$ by \Cref{lem:join}; all three connected components are isomorphic, and any permutation of the components can be realized by choosing componentwise isomorphisms.
  Hence $P(Z)=S_3$.

  If $X\not\cong Y$, then $X'\not\cong Y'$.
  An automorphism of a graph maps connected components to isomorphic connected components, so $H_1$ must be fixed, while the two copies $H_2,H_3$ may be interchanged.
  Hence the induced group is exactly $\langle(23)\rangle$.
\end{proof}

\subsection{Encoding a graph as a binary code}\label{sec:graph-code}

Let $Z$ have $N$ vertices and $M$ edges, and let \[ D_Z\in \F_2^{M\times N} \] be the unoriented edge--vertex incidence matrix: the row indexed by $e=\{u,v\}$ has ones in columns $u,v$.
Following \cite{petrank1997code}, define \[ M_Z=[I_M\mid I_M\mid I_M\mid D_Z] \] and \[ C(Z):=\operatorname{rowspan}(M_Z)\leq \F_2^{3M+N}. \]
The first $3M$ coordinates will be called \emph{private coordinates}; each edge has one private coordinate in each of the three identity blocks.
The final $N$ coordinates are the \emph{vertex coordinates}.

For an edge $e$, let $r_e$ denote the corresponding row of $M_Z$.

\begin{lemma}
  \label{lem:weight5}
  The weight-five codewords of $C(Z)$ are exactly the rows $r_e$, one for each edge $e\in E(Z)$.
\end{lemma}

\begin{proof}
  Each row $r_e$ has one nonzero coordinate in each identity block and two nonzero vertex coordinates, hence $\wt(r_e)=5$.

  Every codeword of $C(Z)$ is the sum of a set $S$ of rows.
  If $|S|=t$, then in each identity block the sum contains exactly the $t$ private coordinates indexed by $S$; since those coordinates are distinct, there is no cancellation.
  Therefore the contribution of the three identity blocks to the total weight is exactly $3t$.
  If $t\geq 2$, this is at least $6$, regardless of the vertex part.
  Hence a weight-five codeword must have $t=1$ and is one of the rows.
\end{proof}

\begin{lemma}
  \label{lem:vertexintrinsic}
  Assume every vertex of $Z$ has degree at least $2$, i.e., $\delta(Z) \ge 2$.
  Then the partition of the coordinates of $C(Z)$ into private coordinates and vertex coordinates is intrinsic to the code.
\end{lemma}

\begin{proof}
  By \Cref{lem:weight5}, the set of weight-five codewords is intrinsic.
  A private coordinate belongs to exactly one such codeword: the edge row associated with that private coordinate.
  A vertex coordinate $v$ belongs to exactly $\deg_Z(v)\geq 2$ weight-five codewords, namely the rows corresponding to incident edges.
  Thus one can distinguish the two types of coordinates solely by counting membership in weight-five codewords.
\end{proof}

\begin{proposition}
  \label{prop:autgraphcode}
  Assume $\delta(Z)\geq2$.
  Restriction to the vertex coordinates defines a surjective homomorphism \[ \nu:\PAut(C(Z))\twoheadrightarrow\Aut(Z). \]
  Its kernel acts trivially on the vertex coordinates and, independently for each edge, permutes the three private coordinates attached to that edge.
  In particular, \[ \ker\nu\cong (S_3)^M. \]
\end{proposition}

\begin{proof}
  Let $\pi\in\PAut(C(Z))$.
  By \Cref{lem:weight5}, $\pi$ permutes the set $\{r_e:e\in E(Z)\}$; by \Cref{lem:vertexintrinsic}, it preserves the set of vertex coordinates.
  Hence $\pi$ induces a permutation $\sigma$ of the vertices and a permutation of the edges.

  The support of $r_e$ on vertex coordinates is precisely the two endpoints of $e$.
  Since coordinate permutations preserve incidence between codewords and coordinates, if $\pi(r_e)=r_{e'}$, then the two vertex coordinates in the support of $r_e$ are carried to the two vertex coordinates in the support of $r_{e'}$.
  Therefore $\sigma$ preserves adjacency and lies in $\Aut(Z)$.
  This defines the homomorphism $\nu$.

  Conversely, let $\sigma\in\Aut(Z)$.
  It induces a permutation $\tau$ of the edge set.
  Permute the vertex coordinates by $\sigma$, and in each of the three identity blocks permute the private coordinates by the same edge permutation $\tau$.
  Then $r_e$ is sent to $r_{\tau(e)}$, so the row space is preserved.
  Thus $\nu$ is surjective.

  Finally, if $\nu(\pi)=1$, then all vertex coordinates are fixed.
  Since the vertex support of $r_e$ is the unordered endpoint pair of $e$, each row $r_e$ must be fixed as a vector.
  The only freedom is therefore to permute its three private support coordinates among themselves.
  These choices are independent across edges because the private supports of distinct rows are disjoint.
  Hence $\ker\nu\cong(S_3)^M$.
\end{proof}

\subsection{The reduction map}\label{sec:GI-reduction-map}

For $i=1,2,3$, let $s_i\in\F_2^{3M+N}$ be the indicator vector of the vertex coordinates belonging to component $H_i$, with all private coordinates equal to zero.
Define \[ L_0:=\Span\{s_1+s_2,\ s_1+s_3\}. \]

\begin{lemma}
  \label{lem:directsum}
  We have \[ \Span\{s_1,s_2,s_3\}\cap C(Z)=\{0\}. \]
  In particular $L_0\cap C(Z)=\{0\}$ and $\dim L_0=2$.
\end{lemma}

\begin{proof}
  Every nonzero codeword of $C(Z)$ is the sum of a nonempty set of edge rows.
  In each of the three identity blocks, such a sum has a nonzero private coordinate for every selected edge.
  Thus every nonzero codeword of $C(Z)$ has nonzero private part.
  By contrast, every vector in $\Span\{s_1,s_2,s_3\}$ has zero private part.
  Hence the intersection is zero.
  The vectors $s_1,s_2,s_3$ have disjoint nonempty supports and are linearly independent, so the two displayed pairwise sums are independent.
\end{proof}

Set \[ W:=C(Z)\oplus L_0, \qquad V_X:=C(Z), \qquad V_Z:=W^\perp. \]
Then $V_X\subseteq W=V_Z^\perp$, and by \Cref{lem:directsum}
\[
  V_Z^\perp/V_X=W/C(Z)\cong L_0\cong\F_2^2.
\]
Hence the quotient dimension is exactly $k=2$.

\begin{lemma}
  \label{lem:PAut-code}
  For the pair $(V_X,V_Z)$ constructed above, \[ \PAut(V_X,V_Z)=\PAut(C(Z)). \]
\end{lemma}

\begin{proof}
  By definition, $\pi\in \PAut(V_X,V_Z)$ if and only if it preserves $V_X=C(Z)$ and $V_Z$.
  Since coordinate permutations preserve the bilinear form, \[ \pi(V_Z)=V_Z\iff \pi(V_Z^\perp)=V_Z^\perp\iff \pi(W)=W. \]
  Thus $\PAut(V_X,V_Z)$ is the subgroup of $\PAut(C(Z))$ preserving $W$.

  Let $\pi\in\PAut(C(Z))$.
  By \Cref{prop:autgraphcode}, its action on vertex coordinates is an automorphism of $Z$, hence permutes the three connected components according to some $\tau\in P(Z)\leq S_3$.
  Therefore \[ \pi(s_i)=s_{\tau(i)}. \]
  The space $L_0$ is the even-weight subspace of the three-dimensional space $\Span\{s_1,s_2,s_3\}$ with respect to the coefficient vector on the basis $(s_1,s_2,s_3)$, and is invariant under all permutations of the three basis vectors.
  Hence $\pi(L_0)=L_0$.
  Since $\pi(C(Z))=C(Z)$, it follows that $\pi(W)=W$.
  Thus every permutation automorphism lies in $\PAut(V_X,V_Z)$, proving equality.
\end{proof}

The quotient $W/C(Z)$ has precisely three nonzero elements, \[ q_{12}=(s_1+s_2)+C(Z),\qquad q_{13}=(s_1+s_3)+C(Z),\qquad q_{23}=(s_2+s_3)+C(Z). \]
A component permutation $\tau\in S_3$ acts by
\[
  q_{ij}\mapsto q_{\tau(i)\tau(j)}.
\]

\begin{lemma}
  \label{lem:S3GL2}
  The permutation action of $S_3$ on the three unordered pairs $\{12,13,23\}$ identifies $S_3$ with $\GLtwo{2}$ acting on the three nonzero vectors of $\F_2^2$.
\end{lemma}

\begin{proof}
  A two-dimensional vector space over $\F_2$ has exactly three nonzero vectors.
  Every invertible linear map permutes them, giving an injective homomorphism \[ \GLtwo{2}\hookrightarrow S_3. \]
  Since
  \[
    |\GLtwo{2}|=(2^2-1)(2^2-2)=6=|S_3|,
  \]
  the map is an isomorphism.
  Under the basis $q_{12},q_{13}$, the third nonzero vector is $q_{23}=q_{12}+q_{13}$, so the component-pair action is exactly this standard action.
\end{proof}

\begin{proposition}
  \label{prop:imagePZ}
  For the constructed pair, \[ A(\PAut(V_X,V_Z))\cong P(Z)\leq S_3\cong\GLtwo{2}. \]
  More precisely, under the identification in \Cref{lem:S3GL2}, the image $A(\PAut(V_X,V_Z))$ equals $P(Z)$.
\end{proposition}

\begin{proof}
  By \Cref{lem:PAut-code}, $\PAut(V_X,V_Z)=\PAut(C(Z))$.
  The kernel of $\nu$ from \Cref{prop:autgraphcode} fixes every vertex coordinate and therefore fixes every $s_i$, hence acts trivially on $W/C(Z)$.
  Thus the quotient action factors through $\nu(\PAut(V_X,V_Z))=\Aut(Z)$ and depends only on the induced component permutation.
  Every component permutation in $P(Z)$ is realized by some graph automorphism and therefore by some permutation automorphism.
  By \Cref{lem:S3GL2}, the resulting linear action on the quotient is exactly the corresponding subgroup of $\GLtwo{2}$.
\end{proof}

\begin{proof}[Proof of \Cref{thm:GI-hard}]
  Given non-empty graphs $(X,Y)$, construct $X',Y',Z$, then $C(Z)$, $L_0$, $W$, and finally the pair $(V_X,V_Z)$ as above.
  By \Cref{prop:imagePZ} and \Cref{lem:component-group}, \[ A(\PAut(V_X,V_Z))=\GLtwo{2} \iff P(Z)=S_3 \iff X\cong Y. \]
  Thus the construction is a correct many-one reduction.

  It remains to check polynomial time.
  The join with $K_2$ adds two vertices and $2|V(X)|+1$ edges.
  The disjoint union $Z$ therefore has size polynomial in the input.
  If $Z$ has $M$ edges and $N$ vertices, the generator matrix $M_Z$ has dimensions $M\times(3M+N)$, hence polynomial size.
  The space $W$ is generated by the rows of $M_Z$ together with two additional rows $s_1+s_2$ and $s_1+s_3$.
  A generator matrix for $V_Z=W^\perp$ is obtained by Gaussian elimination over $\F_2$, in polynomial time.
  The quotient dimension is exactly two by \Cref{lem:directsum}.
  Therefore the reduction is polynomial-time and maps into $\SURJ_2$.
\end{proof}

\begin{remark}[Why this is not an NP-completeness proof]
  The theorem proves GI-hardness, not NP-hardness.
  Graph Isomorphism is not known to be NP-complete.
  Thus the reduction cannot by itself establish NP-hardness or NP-completeness of $\SURJ$.
  This distinction is essential.
  The classical code-equivalence literature similarly uses reductions from Graph Isomorphism to establish GI-hardness without implying NP-completeness.
\end{remark}

\subsection{GI-hardness of logical relabelling}\label{sec:logical-relabelling}

We now consider the weaker requirement that physical permutations realize every permutation of a fixed logical basis, corresponding to all logical SWAP gates.

\begin{definition}[Full logical relabelling]
  \label{def:full-logical-relabelling}
  Let $(V_X,V_Z)$ be a CSS code with logical $X$ quotient $Q_X$ of dimension $k$, and fix an ordered basis $\mathcal E=(e_1,\ldots,e_k)$ of $Q_X$.
  Let $S(\mathcal E)\leq\GL(Q_X)$ be the subgroup that permutes the basis vectors in $\mathcal E$, so $S(\mathcal E)\cong S_k$.
  The code supports \emph{full logical relabelling} with respect to $\mathcal E$ if \[ S(\mathcal E)\leq A(\PAut(V_X,V_Z)). \]
\end{definition}

\begin{theorem}[GI-hardness of logical relabelling]
  \label{thm:logical-relabelling-GI-hard}
  For every fixed $k\geq2$, deciding whether a binary CSS code with $k$ logical qubits and a specified basis of $Q_X$ supports full logical relabelling is GI-hard.
\end{theorem}

\begin{proof}
  We reduce from Graph Isomorphism on non-empty graphs of equal order; unequal orders are immediate no-instances.
  Given $(X,Y)$, set \[ X':=X\vee K_2, \qquad Y':=Y\vee K_2. \]
  By \Cref{lem:join}, both graphs are connected with minimum degree at least two, and $X'\cong Y'$ if and only if $X\cong Y$.

  Form the disjoint union \[ Z:=H_1\sqcup H_2\sqcup\cdots\sqcup H_k, \] where $H_1:=X'$ and $H_2,\ldots,H_k$ are independent copies of $Y'$.
  Construct the incidence code $C(Z)$ as above.
  For each $i\in[k]$, let $s_i$ be the indicator vector of the vertex coordinates in $H_i$, with all private coordinates equal to zero, and set \[ W:=C(Z)\oplus\Span\{s_1,\ldots,s_k\}. \]
  This sum is direct because every nonzero word of $C(Z)$ has a nonzero private part, whereas every vector in $\Span\{s_1,\ldots,s_k\}$ has zero private part.
  Define \[ V_X:=C(Z), \qquad V_Z:=W^\perp. \]
  Then $V_X\subseteq V_Z^\perp=W$, and the logical $X$ quotient is \[ Q_X=V_Z^\perp/V_X=W/C(Z). \]
  It has dimension $k$ with specified basis \[ \mathcal E:=\bigl(s_1+V_X,\ldots,s_k+V_X\bigr). \]

  We claim that the induced action on $Q_X$ is exactly the permutation action of $\Aut(Z)$ on the connected components $H_1,\ldots,H_k$.
  Certainly $\PAut(V_X,V_Z)\leq\PAut(C(Z))$ because $V_X=C(Z)$.
  Conversely, \Cref{prop:autgraphcode} shows that every element of $\PAut(C(Z))$ induces an automorphism of $Z$ and hence permutes its connected components.
  It therefore permutes $s_1,\ldots,s_k$, preserves $W$, and consequently preserves $V_Z=W^\perp$.
  Thus \[ \PAut(V_X,V_Z)=\PAut(C(Z)). \]
  The kernel of the map $\nu:\PAut(C(Z))\twoheadrightarrow\Aut(Z)$ fixes every vertex coordinate and hence every $s_i$, so it acts trivially on $Q_X$.
  Therefore $A(\PAut(V_X,V_Z))$ is precisely the component-permutation group acting on $\mathcal E$.

  If $X\cong Y$, then all $k$ components of $Z$ are isomorphic, and componentwise isomorphisms realize every permutation of them.
  Hence \[ A(\PAut(V_X,V_Z))=S(\mathcal E)\cong S_k. \]
  If $X\not\cong Y$, then $H_1\not\cong H_i$ for every $i\geq2$, so every automorphism of $Z$ fixes $H_1$ and can only permute the other $k-1$ components.
  The induced image therefore fixes $s_1+V_X$ and does not contain $S(\mathcal E)$.
  Thus the constructed code supports full logical relabelling if and only if $X\cong Y$.

  Since $k$ is fixed, the graph $Z$, the generator matrix of $C(Z)$, the CSS spaces, and the specified quotient basis all have size polynomial in the input.
  This completes the reduction.
\end{proof}

\section{Fixed-distance phantom qLDPC codes}\label{sec:constructions}

The main result of this section is the following.

\phantomqldpc*

We prove \Cref{thm:qldpc-fixedD} in two steps.
First, a degree-split Hamming presentation gives a bounded-degree outer CSS family with full logical action and $k=\Theta(\log n)$, but with one sector distance equal to one (\Cref{sec:sparse-outer}).
Second, block concatenation with a fixed one-logical-qubit CSS code preserves the full logical action and bounded sparsity while multiplying both sector distances (\Cref{sec:concat-inner}).

\subsection{A sparse outer family}\label{sec:sparse-outer}

Fix $k\geq2$ and set \[ P:=\F_2^k\setminus\{0\}, \qquad N:=|P|=2^k-1. \]
Coordinates of $\F_2^P$ are indexed by the nonzero vectors $x\in\F_2^k$.
Define the simplex code and its dual Hamming code by \[ S_k:=\{s_a=(a\cdot x)_{x\in P}:a\in\F_2^k\}, \qquad H_k:=S_k^\perp. \]

\begin{proposition}[\cite{macwilliams1977theory,zumbragel2012pseudocodeword}]
  \label{prop:simplex-hamming-facts}
  The following standard facts hold.
  \begin{enumerate}
    \item
          The dimensions are $\dim S_k=k$ and $\dim H_k=N-k$.
    \item
          Every nonzero word of $S_k$ has weight $2^{k-1}$.
    \item
          The code $H_k$ has a basis consisting of weight-three codewords.
  \end{enumerate}
\end{proposition}

\begin{proof}
  For $a\neq0$, the map $x\mapsto a\cdot x$ is a nonzero linear functional on $\F_2^k$, so it equals one on exactly $2^{k-1}$ vectors, all of which lie in $P$.
  Hence $a\mapsto s_a$ is injective, $\dim S_k=k$, every nonzero $s_a$ has weight $2^{k-1}$, and $\dim H_k=N-k$ by orthogonality.

  For the final claim, let $e_1,\ldots,e_k$ be the standard basis of $\F_2^k$, and for each non-unit $x\in P$ choose an index $i(x)$ with $x_{i(x)}=1$ and set
  \[
    b_x:=\delta_x+\delta_{x+e_{i(x)}}+\delta_{e_{i(x)}}.
  \]
  The three coordinates are distinct and sum to zero, so $b_x\in H_k$ and $\wt(b_x)=3$.
  Order the non-unit coordinates by nondecreasing Hamming weight; in the corresponding submatrix, the row $b_x$ has a one in column $x$, while its only possible other nonzero entry is in column $x+e_{i(x)}$, of strictly smaller weight.
  This submatrix is triangular with unit diagonal, so the $N-k$ vectors $b_x$ are independent and therefore form a basis of $H_k$.
\end{proof}

We apply the variable-copying degree-reduction gadget of \cite{hastings2021quantum,hastings2021fiber}, specialized here to the weight-three Hamming presentation.
Set \[ R:=N-k=\dim H_k \] and fix a weight-three basis $\mathcal B=\{b_1,\ldots,b_R\}$ as in \Cref{prop:simplex-hamming-facts}.
For each $x\in P$, let \[ I_x:=\{j\in[R]:x\in\supp(b_j)\} \] and choose an injection \[ \rho_x:I_x\hookrightarrow[R]. \]
Replace every coordinate $x$ by $R$ copies, giving
\[
  \Omega:=P\times[R], \qquad M:=|\Omega|=NR=(2^k-1)(2^k-1-k).
\]
For $x\in P$ and $1\leq r<R$, define the equality check \[ e_{x,r}:=\delta_{(x,r)}+\delta_{(x,r+1)}. \]
If $b_j=\delta_x+\delta_y+\delta_z$, define its split lift by
\[
  \widehat b_j:= \delta_{(x,\rho_x(j))} +\delta_{(y,\rho_y(j))} +\delta_{(z,\rho_z(j))}.
\]
Define the outer CSS pair by
\[
  V_X^{\mathrm{out}}:=0, \qquad V_Z^{\mathrm{out}}:=\Span\bigl(\{e_{x,r}:x\in P,\ 1\leq r<R\}\cup\{\widehat b_j:j\in[R]\}\bigr).
\]

\begin{theorem}[Sparse outer phantom family]
  \label{thm:qldpc-outer}
  The pair $(V_X^{\mathrm{out}},V_Z^{\mathrm{out}})$ is a phantom CSS code encoding $k$ logical qubits.
  The displayed generators have maximum row weight and column degree at most three.
  Moreover,
  \[
    (V_Z^{\mathrm{out}})^\perp =\operatorname{Rep}_R(S_k) :=\{c\in\F_2^\Omega:c_{(x,r)}=(s_a)_x\text{ for some }a\in\F_2^k\text{ and all }x,r\},
  \]
  and the code satisfies
  \[
    d_X=R2^{k-1}, \qquad d_Z=1, \qquad A\bigl(\PAut(V_X^{\mathrm{out}},V_Z^{\mathrm{out}})\bigr)=\GLtwo{k}.
  \]
  Its length and logical dimension satisfies \[ M=\Theta(4^k), \qquad k=\Theta(\log M). \]
\end{theorem}

\begin{proof}
  Every equality check has weight two and every split Hamming check has weight three.
  A copied coordinate lies in at most two equality checks and, by injectivity of $\rho_x$, at most one split Hamming check.
  Hence the displayed presentation has maximum row weight and column degree at most three.

  If $c\in(V_Z^{\mathrm{out}})^\perp$, orthogonality to the equality checks makes $c$ constant on the $R$ copies of every $x$, say $c_{(x,r)}=a_x$.
  Orthogonality to every $\widehat b_j$ then gives $\langle a,b_j\rangle=0$.
  Since the $b_j$ form a basis of $H_k$, we have $a\in H_k^\perp=S_k$.
  Conversely, every repeated simplex word is orthogonal to all displayed generators, proving the claimed description of $(V_Z^{\mathrm{out}})^\perp$.

  The space $(V_Z^{\mathrm{out}})^\perp$ has dimension $k$, so the pair encodes $k$ qubits because $V_X^{\mathrm{out}}=0$.
  Its logical $X$ quotient is $(V_Z^{\mathrm{out}})^\perp$, whose nonzero words have weight $R2^{k-1}$ by \Cref{prop:simplex-hamming-facts}, giving the stated $d_X$.
  For every $(x,r)\in\Omega$, some repeated simplex word is nonzero at $(x,r)$, so $\delta_{(x,r)}\notin V_Z^{\mathrm{out}}$ and hence $d_Z=1$.

  For $T\in\GLtwo{k}$, the coordinate permutation \[ \pi_T:(x,r)\longmapsto(Tx,r) \] sends the repeated word $s_a$ to $s_{T^{-\top}a}$.
  It therefore preserves $(V_Z^{\mathrm{out}})^\perp$, and hence preserves $V_Z^{\mathrm{out}}$, while fixing $V_X^{\mathrm{out}}=0$.
  On the logical $X$ quotient these permutations realize all of $\GLtwo{k}$, so the pair is phantom.
  Finally, $M=NR=\Theta(4^k)$.
\end{proof}

\begin{remark}
  The injections $\rho_x$ need not respect the $\GLtwo{k}$ action, so the sparse generators themselves need not be permuted by every $\pi_T$.
  Phantomness is nevertheless preserved because the permutations preserve the stabilizer space through its orthogonal complement.
\end{remark}

\subsection{Concatenation with a fixed inner code}\label{sec:concat-inner}

The outer family already has full logical action and bounded sparsity, but its $Z$-distance is one.
We repair this single defect by concatenating every outer coordinate with the same fixed one-logical-qubit CSS code.
This is the standard concatenated-code construction of \cite{knill1996concatenated}, specialized to CSS codes with a one-logical-qubit inner block.
In the phantom setting, the product parameters and preservation of outer-level phantomness are stated explicitly in \cite{koh2026entangling}.

Let $U_X,U_Z\leq\F_2^m$ satisfy
\[
  U_X\subseteq U_Z^\perp, \qquad m-\dim U_X-\dim U_Z=1.
\]
Choose minimum-weight logical representatives
\[
  x\in U_Z^\perp\setminus U_X, \qquad z\in U_X^\perp\setminus U_Z,
\]
and set
\[
  \delta_X:=\wt(x), \qquad \delta_Z:=\wt(z).
\]
The perfect pairing between the one-dimensional inner logical quotients gives $\langle x,z\rangle=1$.

Let $(V_X,V_Z)$ be an outer CSS pair of length $n_0$.
Replace each outer coordinate by a block of $m$ coordinates and define
\[
  L_X(a_1,\ldots,a_{n_0}):=(a_1x,\ldots,a_{n_0}x), \qquad L_Z(a_1,\ldots,a_{n_0}):=(a_1z,\ldots,a_{n_0}z).
\]
The concatenated CSS spaces are
\[
  \widetilde V_X:=U_X^{\oplus n_0}+L_X(V_X), \qquad \widetilde V_Z:=U_Z^{\oplus n_0}+L_Z(V_Z).
\]

\begin{proposition}[Standard block concatenation]
  \label{prop:block-concatenation}
  The concatenated pair has the following properties.
  \begin{enumerate}[label=(\roman*)]
    \item
          It is a CSS pair with the same logical dimension as $(V_X,V_Z)$.
    \item
          Its logical quotients identify naturally with those of the outer pair.
    \item
          Its distances satisfy \[ \widetilde d_X=\delta_Xd_X, \qquad \widetilde d_Z=\delta_Zd_Z. \]
    \item
          Every element of $\PAut(V_X,V_Z)$ lifts by permuting the inner blocks, with the same induced logical action.
    \item
          If the chosen inner and outer generating sets have row-weight bounds $w_{\mathrm{in}},w_{\mathrm{out}}$ and column-degree bounds $\Delta_{\mathrm{in}},\Delta_{\mathrm{out}}$, then the natural concatenated generators satisfy
          \[
            \widetilde w\leq\max\{w_{\mathrm{in}},\delta_Xw_{\mathrm{out}},\delta_Zw_{\mathrm{out}}\}, \qquad \widetilde \Delta\leq\Delta_{\mathrm{in}}+\Delta_{\mathrm{out}}.
          \]
  \end{enumerate}
\end{proposition}

\begin{proof}
  Items (i)--(iii) are the standard CSS specialization of quantum-code concatenation \cite{knill1996concatenated}: the one-dimensional inner logical quotients label each block by a bit, thereby identifying the concatenated logical quotients with the outer ones, and every occupied block contributes minimum weight $\delta_X$ or $\delta_Z$.
  For (iv), an outer coordinate permutation lifts by permuting whole inner blocks, preserves the concatenated spaces, and acts on their quotient labels exactly as it acts on the outer code \cite{koh2026entangling}.
  For (v), an internal generator retains weight at most $w_{\mathrm{in}}$, while a lifted outer generator has its weight multiplied by $\delta_X$ or $\delta_Z$.
  A physical coordinate belongs to at most $\Delta_{\mathrm{in}}$ internal generators and to lifts of at most $\Delta_{\mathrm{out}}$ outer generators, proving the sparsity bounds.
\end{proof}

\begin{proof}[Proof of \Cref{thm:qldpc-fixedD}]
  Fix $D\geq1$ and choose a one-logical-qubit CSS code with \[ \delta_X,\delta_Z\geq D. \]
  Such a code exists for every $D$, for example as a sufficiently large planar surface-code block \cite{dennis2002topological}.
  Because $D$ is fixed, its length, generator weights, column degree, and chosen logical representative weights are constants depending only on $D$.

  For each $k\geq2$, concatenate this inner code with the outer pair of \Cref{thm:qldpc-outer}.
  By \Cref{prop:block-concatenation}, the resulting CSS code encodes $k$ qubits, remains phantom, and has bounded row and column degrees depending only on $D$.
  Its sector distances satisfy
  \[
    \widetilde d_X=\delta_XR2^{k-1}\geq D, \qquad \widetilde d_Z=\delta_Z\geq D.
  \]
  If the fixed inner block has length $m_D$, then the total length is \[ n=m_D(2^k-1)(2^k-1-k)=\Theta_D(4^k), \] and hence $k=\Theta(\log n)$.
\end{proof}

The Steane code gives a concrete instance of the construction at distance three.

\begin{corollary}[Explicit distance-three family]
  \label{cor:qldpc-steane}
  For every $k\geq2$ there is a binary phantom qLDPC CSS code with
  \[
    n=7(2^k-1)(2^k-1-k), \qquad K=k, \qquad d=3,
  \]
  and a stabilizer presentation with maximum row weight and column degree at most nine.
  More precisely, \[ d_X=3(2^k-1-k)2^{k-1}, \qquad d_Z=3. \]
\end{corollary}

\begin{proof}
  Take $U_X=U_Z$ to be the $[7,3,4]$ simplex code and choose a weight-three representative of its unique logical class, giving the Steane $[[7,1,3]]$ CSS code \cite{steane1996error}.
  Its standard presentation has row weight at most four and combined column degree at most six, while the chosen logical representatives have weight three.
  Applying \Cref{prop:block-concatenation} to the outer presentation of \Cref{thm:qldpc-outer} gives the stated distances and the bounds \[ w\leq\max\{4,3\cdot3\}=9, \qquad \Delta\leq6+3=9. \]
\end{proof}

\section{Distance no-go theorems}\label{sec:obstructions}

We now prove the complementary obstruction for arbitrary binary phantom stabilizer codes.

\stabilizernogo*

In \Crefrange{sec:sunflower}{sec:css-no-go}, we first prove the obstruction for CSS codes, where the separate logical $X$- and $Z$-quotients make the underlying structure clear.
Then \Cref{sec:stabilizer-no-go} extends the argument to general stabilizer codes and completes the proof of \Cref{thm:stabilizer-no-go}.

Throughout the CSS part of this section, fix an integer $W\geq1$ and suppose \[ V_X\subseteq V_Z^\perp\leq\F_2^n, \qquad w(V_X,V_Z)\leq W, \qquad A(\PAut(V_X,V_Z))=\GLtwo{k}. \]
For brevity, write $\PAut:=\PAut(V_X,V_Z)$ throughout this part.
Write $d_X,d_Z,d$ as in \Cref{sec:setup}.

\subsection{Sunflowers and canonical low-weight atoms}\label{sec:sunflower}

We begin with the classical sunflower lemma in a form sufficient for our purposes \cite{erdos1960intersection}.

\begin{lemma}[Sunflower lemma]
  \label{lem:sunflower}
  Let $s\geq0$ and $r\geq2$.
  If a finite family $\mathcal F$ of distinct sets, each of cardinality at most $s$, has more than \[ s!
    (r-1)^s \]
  members, then it contains distinct sets $S_1,\ldots,S_r$ and a set $C$ such that
  \[
    S_a=C\sqcup P_a
  \]
  with pairwise disjoint petals $P_1,\ldots,P_r$.
\end{lemma}

\begin{proof}
  Pad every $S\in\mathcal F$ to cardinality exactly $s$ by adjoining $s-|S|$ new dummy elements that occur in no other padded set.
  An $r$-sunflower among the padded sets has no dummy element in its common core, so deleting all dummies gives an $r$-sunflower in the original family.
  It therefore suffices to treat an $s$-uniform family.

  We argue by induction on $s$.
  The case $s=0$ is vacuous because there is only one distinct empty set.
  For $s\geq1$, choose a maximal pairwise disjoint subfamily $\mathcal M\subseteq\mathcal F$.
  If $|\mathcal M|\geq r$, any $r$ members form a sunflower with empty core.
  Otherwise \[ X:=\bigcup_{M\in\mathcal M}M \] has at most $s(r-1)$ elements.
  Maximality implies that every member of $\mathcal F$ meets $X$.
  Hence some $x\in X$ belongs to more than \[ \frac{s!
      (r-1)^s}{s(r-1)}=(s-1)!(r-1)^{s-1} \]
  members of $\mathcal F$.
  Delete $x$ from all those sets.
  The resulting sets remain distinct and are $(s-1)$-uniform, so induction gives an $r$-sunflower.
  Restoring $x$ adds it to the common core.
\end{proof}

For a nonzero binary word $a$, its support determines the word uniquely.
This lets us define an intrinsic set of minimal low-weight generators.

\begin{definition}[Atom]
  \label{def:atom}
  For $i\in\{X,Z\}$, a nonzero word $a\in V_i$ of weight at most $W$ is an \emph{atom} if no nonzero word of $V_i$ has support strictly contained in $\supp(a)$.
  Let $\mathcal A_i$ be the set of all atoms of $V_i$ having weight at most $W$.
\end{definition}

We call $\mathcal A_i$ canonical because it is determined solely by $V_i$ and support containment, without any choice of basis or generating set.
Equivalently, $a$ is atomic if it cannot be written as a sum $a=b+c$ of two nonzero codewords of $V_i$ with disjoint supports.

\begin{lemma}[Atomic incidence bound]
  \label{lem:atomic-incidence}
  Assume $d_X,d_Z>W$.
  Then each $\mathcal A_i$ is a canonical $\PAut$-invariant spanning set of $V_i$, every atom has weight at most $W$, and every physical coordinate belongs to the supports of at most \[ F(W):=W! W^W \] atoms from $\mathcal A_i$.
\end{lemma}

\begin{proof}
  Take a basis of $V_i$ whose words have weight at most $W$.
  Whenever a basis word is not atomic, split it as a disjoint sum of two nonzero words of strictly smaller support.
  Iterating terminates and expresses every basis word as a sum of atoms of weight at most $W$.
  Hence $\mathcal A_i$ spans $V_i$.
  The definition depends only on the subspace and support containment, so every coordinate permutation preserving $V_i$ permutes $\mathcal A_i$.

  Fix a coordinate $q$ and suppose more than $F(W)$ atoms of $\mathcal A_Z$ contain $q$.
  By \Cref{lem:sunflower} with $s=W$ and $r=W+1$, there are distinct atomic supports \[ S_a=C\sqcup P_a, \qquad 1\leq a\leq W+1, \] whose petals are pairwise disjoint.
  Since $q\in C$, the core is nonempty.

  Choose a basis $B_X$ of $V_X$ whose words have weight at most $W$.
  For any $x\in B_X$, at least one of the $W+1$ petals is disjoint from $\supp(x)$: a weight-$W$ vector cannot meet all $W+1$ pairwise disjoint nonempty petals, and an empty petal is automatically missed.
  If $P_a$ is missed, then, writing $\mathbf 1_S$ for the binary indicator of a support $S$, \[ 0=\langle x,\mathbf 1_{S_a}\rangle =\langle x,\mathbf 1_C\rangle+\langle x,\mathbf 1_{P_a}\rangle =\langle x,\mathbf 1_C\rangle, \] because $\mathbf 1_{S_a}\in V_Z\subseteq V_X^\perp$.
  Applying the same equality to each $S_b$ gives \[ \langle x,\mathbf 1_{P_b}\rangle=0 \qquad\text{for all }b. \]
  Since this holds for every $x\in B_X$,
  \[
    \mathbf 1_C,\mathbf 1_{P_b}\in V_X^\perp.
  \]
  Now $0<|C|\leq W<d_Z$, so $\mathbf 1_C$ cannot represent a nontrivial $Z$-logical class and therefore lies in $V_Z$.
  At least one petal, say $P_b$, is nonempty because the supports $S_a$ are distinct.
  Again $|P_b|\leq W<d_Z$ implies $\mathbf 1_{P_b}\in V_Z$.
  Consequently \[ \mathbf 1_{S_b}=\mathbf 1_C+\mathbf 1_{P_b} \] is a disjoint sum of two nonzero words of $V_Z$, contradicting atomicity.
  Thus at most $F(W)$ atoms of $\mathcal A_Z$ contain $q$.

  Interchanging $V_Z$ and $V_X$ and using $d_X>W$ proves the same bound for $\mathcal A_X$.
\end{proof}

\subsection{The unique logically active incidence component}

Form the sector-labelled bipartite incidence graph $\Gamma$: its vertices are the $n$ physical coordinates and the atoms in $\mathcal A_X\sqcup\mathcal A_Z$, and an atom is adjacent to every coordinate in its support.
The sector label records whether an atom belongs to $\mathcal A_X$ or $\mathcal A_Z$.

\begin{lemma}[Active component]
  \label{lem:active-component}
  Assume $d_X,d_Z>W$ and $k>0$.
  Then precisely one connected component $\Gamma_*$ of $\Gamma$ has nonzero logical quotient.
  This component is $\PAut$-invariant, and the image \[ H\leq\Aut(\Gamma_*) \] of the restriction map from $\PAut$ has a quotient isomorphic to $\GLtwo{k}$.
  Moreover, \[ |V(\Gamma)|\leq (1+2F(W))n, \qquad \Delta(\Gamma)\leq \Delta(W):=\max\{W,2F(W)\}. \]
\end{lemma}

\begin{proof}
  Every atom has degree at most $W$.
  By \Cref{lem:atomic-incidence}, each coordinate is incident with at most $F(W)$ atoms of either sector and hence has degree at most $2F(W)$.
  The total number of atom vertices is at most the total number of incidences, which is at most $2F(W)n$.
  This proves the size and degree bounds.

  For a connected component $C$ and $i\in\{X,Z\}$, let $I_C$ be its physical-coordinate set and let $V_{i,C}\leq\F_2^{I_C}$ be the span of the sector-$i$ atoms in $C$.
  Since the atoms span $V_i$ and different graph components have disjoint coordinate sets, \[ V_i=\bigoplus_C V_{i,C}. \]
  Taking orthogonal complements component-wise gives
  \[
    Q_X=V_Z^\perp/V_X =\bigoplus_C Q_C, \qquad Q_C:=V_{Z,C}^\perp/V_{X,C},
  \]
  where the orthogonal complement in $Q_C$ is taken inside $\F_2^{I_C}$.

  The atom sets $\mathcal A_X$ and $\mathcal A_Z$ are $\PAut$-invariant, so $\PAut$ acts on the sector-labelled graph $\Gamma$.
  The group $\PAut$ permutes graph components and therefore permutes the family of nonzero subspaces $Q_C$.
  Because the induced action is all of $\GL(Q_X)$, and this full linear group must preserve that direct-sum system, there cannot be two nonzero summands.
  Indeed, if $x\in Q_C$ and $y\in Q_D$ are nonzero with $C\neq D$, extend $x,y$ to a basis of $Q_X$ and take the invertible linear map sending $x$ to $x+y$ while fixing the other basis vectors.
  It does not permute the summands $Q_C$, therefore cannot be an induced action, contradicting surjectivity.
  Hence there is a unique component $\Gamma_*$ with $Q_*=Q_X\neq0$, and uniqueness makes it $\PAut$-invariant.

  Let $H$ be the image of $\PAut\to\Aut(\Gamma_*)$.
  If two elements of $\PAut$ have the same restriction to $\Gamma_*$, they induce the same action on $Q_*=Q_X$; every other component has zero quotient.
  Therefore the quotient action $\PAut\twoheadrightarrow\GLtwo{k}$ factors through a surjection \[ H\twoheadrightarrow\GLtwo{k}. \]
\end{proof}

\subsection{Bounded degree forces large quotient groups into large vertex orbits}

For a positive integer $m$ and a threshold $\Delta$, write $m_{>\Delta}$ for the largest divisor of $m$ all of whose prime factors are greater than $\Delta$.
For a finite group $G$ acting on a set $\Omega$ and $x\in\Omega$, write $G_x:=\{g\in G:g(x)=x\}$ for the stabilizer of $x$ and $Gx:=\{g(x):g\in G\}$ for its orbit.
The orbit-stabilizer theorem gives $|Gx| = [G:G_x]$.

\begin{lemma}[Prime support of a vertex stabilizer]
  \label{lem:bounded-degree-quotient}
  Let $\Gamma$ be a finite connected graph of maximum degree at most $\Delta\geq2$, let $H\leq\Aut(\Gamma)$, and suppose $\phi:H\twoheadrightarrow S$ is surjective.
  Then for every vertex $v$:
  \begin{enumerate}[label=(\roman*)]
    \item
          every prime divisor of $|H_v|$ is at most $\Delta$;
    \item
          the orbit size satisfies \[ |Hv|\geq |S|_{>\Delta}. \]
  \end{enumerate}
\end{lemma}

\begin{proof}
  Let $B_i(v)$ be the radius-$i$ ball around $v$, and let $K_i$ be the pointwise stabilizer of $B_i(v)$ in $H_v$.
  Thus $K_0=H_v$, and $K_i=1$ once the ball is the whole connected graph.
  If $S_i(v)=B_i(v)\setminus B_{i-1}(v)$, then an element of $K_i$ fixes every $u\in S_i(v)$ and permutes its neighbors in $S_{i+1}(v)$.
  Recording these permutations gives an injection \[ K_i/K_{i+1}\hookrightarrow \prod_{u\in S_i(v)} \operatorname{Sym}\bigl(\mathcal N(u)\cap S_{i+1}(v)\bigr). \]
  Every symmetric group on the right has degree at most $\Delta$, so every prime divisor of every section $K_i/K_{i+1}$ is at most $\Delta$.
  The same is therefore true of $H_v$.

  Put $J=\phi(H_v)$.
  It is a quotient of $H_v$, so $|J|$ has no prime divisor greater than $\Delta$.
  Since $\phi$ is surjective, \[ |Hv|=[H:H_v] =[S:J]\,[\ker\phi:\ker\phi\cap H_v]. \]
  For each prime $p>\Delta$, the full $p$-primary part of $|S|$ divides $[S:J]$, and hence divides $|Hv|$.
  Taking the product over all such primes gives $|Hv|\geq |S|_{>\Delta}$.
\end{proof}

It remains to estimate the large-prime part of $|\GLtwo{k}|$.
For a prime $p$ and an integer $a$ with $p\nmid a$, let $\operatorname{ord}_p(a)$ denote the least positive integer $r$ such that $a^r\equiv1\pmod p$.
For a nonzero integer $m$, let $v_p(m):=\max\{e\geq0:p^e\mid m\}$ denote its $p$-adic valuation.

\begin{lemma}[Large-prime part of $\GLtwo{k}$]
  \label{lem:large-prime-part}
  For every fixed $\Delta\geq2$ there is a constant $C_\Delta$ such that \[ |\GLtwo{k}|_{>\Delta} \geq 2^{k(k-1)/2-C_\Delta k} \] for all $k\geq1$.
\end{lemma}

\begin{proof}
  The order formula can be written as \[ |\GLtwo{k}| =2^{k(k-1)/2}\prod_{i=1}^k(2^i-1). \]
  Set \[ P_k:=\prod_{i=1}^k(2^i-1). \]
  Since $2^i-1\geq2^{i-1}$, \[ P_k\geq2^{k(k-1)/2}. \]
  We show that the part of the odd number $P_k$ supported on the finitely many odd primes $p\leq\Delta$ is only $2^{O_\Delta(k)}$.

  Fix an odd prime $p\leq\Delta$ and let $r=\operatorname{ord}_p(2)$.
  A factor $2^i-1$ is divisible by $p$ exactly when $r\mid i$.
  The standard lifting-the-exponent calculation gives \[ v_p(2^{rm}-1)=v_p(2^r-1)+v_p(m) \qquad(m\geq1). \]
  With $M=\lfloor k/r\rfloor$, \[ v_p(P_k) =Mv_p(2^r-1)+v_p(M!) =O_p(k). \]
  Summing over the finitely many odd primes $p\leq\Delta$ shows that the base-two logarithm of the $\leq\Delta$ prime part of $P_k$ is $O_\Delta(k)$.
  Removing it leaves a divisor supported only on primes greater than $\Delta$ of size at least \[ 2^{k(k-1)/2-C_\Delta k}. \]
  This divisor also divides $|\GLtwo{k}|$, proving the claim.
\end{proof}

\subsection{The CSS form of the no-go theorem}\label{sec:css-no-go}

We can now combine the three ingredients.

\begin{lemma}[CSS length lowerbound]
  \label{lem:css-n-lowerbound}
  For every integer $W\geq1$ there is a constant $C_W$ such that the following holds.
  If
  \[
    w(V_X,V_Z)\leq W, \qquad A(\PAut(V_X,V_Z))=\GLtwo{k}, \qquad d>W,
  \]
  then
  \begin{equation}
    \label{eq:quant-nogo}
    (1+2W! W^W)n \geq 2^{k(k-1)/2-C_Wk}.
  \end{equation}
  In particular, \[ \log n=\Omega_W(k^2). \]
\end{lemma}

\begin{proof}
  Because $d=\min\{d_X,d_Z\}>W$, \Cref{lem:atomic-incidence} applies.
  By \Cref{lem:active-component}, there is a connected graph component $\Gamma_*$ such that \[ |V(\Gamma_*)|\leq(1+2W! W^W)n, \qquad \Delta(\Gamma_*)\leq\Delta(W), \] and a subgroup $H\leq\Aut(\Gamma_*)$ admitting a surjection \[ H\twoheadrightarrow\GLtwo{k}. \]
  Apply \Cref{lem:bounded-degree-quotient} to any vertex $v\in V(\Gamma_*)$ and then \Cref{lem:large-prime-part} with $\Delta=\Delta(W)$.
  Since the orbit $Hv$ is contained in $V(\Gamma_*)$, \[ (1+2W! W^W)n \geq |V(\Gamma_*)| \geq |Hv| \geq |\GLtwo{k}|_{>\Delta(W)} \geq 2^{k(k-1)/2-C_{\Delta(W)}k}. \]
  Renaming the last constant as $C_W$ gives \eqref{eq:quant-nogo}.
\end{proof}

The CSS argument already gives the promised distance ceiling on this subclass.

\begin{theorem}[CSS distance no-go]
  \label{thm:css-no-go}
  Let \[ \{(V_X^{(t)},V_Z^{(t)})\}_{t\geq1} \] be an unbounded family of binary CSS pairs with lengths $n_t$, logical dimensions $k_t$, generator weights $w_t$, and distances $d_t$.
  Suppose \[ A(\PAut(V_X^{(t)},V_Z^{(t)}))=\GLtwo{k_t}, \qquad k_t=\omega(\sqrt{\log n_t}), \qquad w_t=O(1). \]
  Then, for all sufficiently large $t$,
  \[
    d_t\leq w_t.
  \]
  In particular, every bounded-check-weight phantom family with $k=\omega(\sqrt{\log n})$ has bounded distance.
\end{theorem}

\begin{proof}
  Choose $W$ with $w_t\leq W$ for all sufficiently large $t$.
  Since the family is unbounded and $k_t=\omega(\sqrt{\log n_t})$, we have $k_t\to\infty$.

  Suppose, for contradiction, that $d_t>w_t$ for infinitely many $t$.
  The integer $w_t$ takes values in the finite set $\{0,1,\ldots,W\}$, so after passing to an infinite subsequence it is constant, say $w_t=W_0$.
  The case $W_0=0$ cannot occur on an unbounded subsequence: then $V_X^{(t)}=V_Z^{(t)}=0$, so $k_t=n_t$, contradicting $k_t=\omega(\sqrt{\log n_t})$.
  Thus $W_0\geq1$.

  Along this subsequence, $d_t>W_0$.
  By \Cref{lem:css-n-lowerbound}, \[ \log_2 n_t \geq \frac{k_t(k_t-1)}2-C_{W_0}k_t-O_{W_0}(1) =\Omega_{W_0}(k_t^2). \]
  But $k_t=\omega(\sqrt{\log n_t})$ also gives $\log n_t=o(k_t^2)$, a contradiction as $k_t\to\infty$.
  Hence only finitely many members can satisfy $d_t>w_t$.
\end{proof}

\begin{remark}[Where logical scaling enters]
  The local combinatorial part of the proof does not use an asymptotic hypothesis on $k$.
  The stronger finite-length statement is \Cref{lem:css-n-lowerbound}: whenever $d>W$, one must pay $\log n=\Omega_W(k^2)$.
  The hypothesis $k=\omega(\sqrt{\log n})$ is used only to contradict this lower bound.
\end{remark}

\begin{remark}[Relation to qLDPC]
  The theorem assumes only a bounded-weight basis for each stabilizer space.
  It therefore applies a fortiori to standard qLDPC families, which additionally require bounded qubit degree.
  No bounded-degree hypothesis is needed for the no-go.
  In particular, the qLDPC families of \Cref{cor:qldpc-steane} and \Cref{thm:qldpc-fixedD} lie above the square-root-logarithmic threshold and necessarily remain at constant distance, exactly as their construction does.
\end{remark}

\subsection{General stabilizer codes}\label{sec:stabilizer-no-go}

We now prove the obstruction in the full setting of \Cref{def:phantom-stabilizer}.
The preceding CSS argument remains in the paper because its unique-active-component step is particularly easy to understand: it uses only the fact that the natural $\GLtwo{k}$-module has no invariant direct-sum system.
For an arbitrary stabilizer code, the corresponding step uses the symplectic logical module $V\oplus V^*$ and the following observation.

\begin{lemma}[Invariant logical subspaces]
  \label{lem:stabilizer-invariant-subspaces}
  Let $k=\dim V\geq3$ and let $\GL(V)$ act on $L=V\oplus V^*$ through $\rho$ (see \Cref{def:phantom-stabilizer}).
  The only invariant subspaces of $L$ are
  \[
    0,\qquad V\oplus0,\qquad 0\oplus V^*,\qquad L.
  \]
  The middle two are Lagrangian and, in particular, degenerate under the restricted symplectic form.
\end{lemma}

\begin{proof}
  The natural modules $V$ and $V^*$ are irreducible: the orbit of any nonzero vector spans the whole module.
  They are not isomorphic as $\GL(V)$-modules when $k\geq3$.
  Indeed, an equivariant isomorphism $T:V\to V^*$ would define a nonzero bilinear form $B(x,y):=T(x)(y)$ satisfying
  \[
    B(Fx,Fy)=T(Fx)(Fy)=\bigl(T(x)\circ F^{-1}\bigr)(Fy)=B(x,y), \qquad \text{for } F\in\GL(V).
  \]
  In a basis $\{e_1,\dots,e_k\}\subset V$, invariance under coordinate permutations would force
  \[
    B(e_i,e_i)=B(e_1,e_1)=:a, \qquad B(e_i,e_j)=B(e_1,e_2)=:b, \qquad \text{for all $i$ and $j\neq i$}.
  \]
  For the transvection $\tau(e_1)=e_1+e_2$ fixing the other basis vectors, invariance gives
  \[
    a=B(e_1,e_1)=B(e_1+e_2,e_1+e_2)=a+b+b+a=0
  \]
  and, since $k\geq3$,
  \[
    b=B(e_1,e_3)=B(e_1+e_2,e_3)=b+b=0.
  \]
  Thus $B=0$, contradicting that $T$ is an isomorphism, and hence $V\not\cong V^*$ as $\GL(V)$-modules.

  Now let $U\leq V\oplus V^*$ be invariant, and denote the coordinate projections by $\pi_V$ and $\pi_{V^*}$.
  By irreducibility, each image is either zero or the entire corresponding module.
  If one image is zero, irreducibility gives one of $0$, $V\oplus0$, or $0\oplus V^*$.
  Suppose instead that both images are surjective.
  The projection kernels are invariant and satisfy
  \[
    \ker(\pi_V|_U)=U\cap(0\oplus V^*)\in\{0,0\oplus V^*\}, \qquad \ker(\pi_{V^*}|_U)=U\cap(V\oplus0)\in\{0,V\oplus0\}.
  \]
  If either kernel is nonzero, $U$ contains the corresponding summand, and surjectivity of the other projection then gives $U=L$.
  If both kernels vanish, both projections are isomorphisms and
  \[
    T:=\pi_{V^*}|_U\circ(\pi_V|_U)^{-1}:V\longrightarrow V^*
  \]
  is a $\GL(V)$-equivariant isomorphism, contradicting $V\not\cong V^*$.
  This proves the list.

  Finally, for $x,y\in V$ and $f,g\in V^*$, \cref{eq:canonical-logical-symplectic-form} gives
  \[
    \omega((x,0),(y,0))=0, \qquad \omega((0,f),(0,g))=0.
  \]
  Thus the restricted form on either middle subspace is identically zero.
  Each has dimension $k=\frac12\dim L$, so both are maximal isotropic, equivalently Lagrangian.
\end{proof}

General Pauli words with the same support can have different $X$, $Y$, and $Z$ labels.
The sunflower step must therefore remember the labels on the common core.

\begin{definition}[Stabilizer atom]
  A nonzero $a\in S$ of weight $\wt_q(a) \le W$ is a \emph{stabilizer atom} if it cannot be written as $a=b+c$ with nonzero $b,c\in S$ such that $\supp_q(b) \cap \supp_q(c) = \varnothing$.
  Let $\mathcal A(S,W)$ denote the set of stabilizer atoms of weight at most $W$.
\end{definition}

\begin{lemma}[Stabilizer atomic incidence bound]
  \label{lem:stabilizer-atomic-incidence}
  Suppose $w(S)\leq W$ and $d(S)>W$.
  Then $\mathcal A(S,W)$ is a canonical $\PAut(S)$-invariant spanning set of $S$, and every physical coordinate belongs to the supports of at most
  \begin{equation}
    \label{eq:stabilizer-sunflower-constant}
    F_{\mathrm{stab}}(W) :=3^W W!\bigl((W+1)3^W-1\bigr)^W.
  \end{equation}
  atoms from $\mathcal A(S,W)$.
\end{lemma}

\begin{proof}
  The spanning, canonicity, and $\PAut(S)$-invariance assertions follow from the recursive-splitting argument, just as in the proof of \Cref{lem:atomic-incidence}.

  Fix a coordinate $q$.
  A fixed support of size at most $W$ carries at most $3^W$ nonzero Pauli labellings.
  If more than $F_{\mathrm{stab}}(W)$ atoms contain $q$, their supports therefore include more than
  \[
    W!\bigl((W+1)3^W-1\bigr)^W
  \]
  distinct sets.
  By \Cref{lem:sunflower}, these contain a sunflower of $(W+1)3^W$ supports
  \[
    C\sqcup P_j,
  \]
  where $q\in C$ and the petals are pairwise disjoint.
  There are at most $3^{|C|}\leq3^W$ possible Pauli restrictions to $C$.
  We may consequently choose $W+1$ atoms $a_j$ with the same nonzero core restriction $c$ and write
  \[
    a_j=c+p_j, \qquad \supp_q(p_j)=P_j.
  \]

  From this point, the sunflower contradiction proceeds exactly as in the proof of \Cref{lem:atomic-incidence}: a weight-$W$ basis of $S$ and isotropy imply that $c,p_i\in S^{\perp_s}$, while $W<d(S)$ forces every nonzero such word into $S$, contradicting atomicity for any nonempty petal.
\end{proof}

Form the bipartite incidence graph $\Gamma(S,W)$ whose vertices are the physical coordinates and the atoms in $\mathcal A(S,W)$, with adjacency given by support incidence.

\begin{lemma}[Active stabilizer component]
  \label{lem:active-stabilizer-component}
  Let $S$ be a phantom stabilizer code with $k\geq3$, $w(S)\leq W$, and $d(S)>W$.
  Precisely one connected component $\Gamma_*$ of $\Gamma(S,W)$ has nonzero logical Pauli space.
  It is invariant under the subgroup of $\PAut(S)$ realizing the standard logical $\GLtwo{k}$, and the image $H\leq\Aut(\Gamma_*)$ of the restriction map has a quotient isomorphic to $\GLtwo{k}$.
  Moreover,
  \[
    |V(\Gamma)|\leq(1+F_{\mathrm{stab}}(W))n, \qquad \Delta(\Gamma)\leq \Delta_{\mathrm{stab}}(W):=\max\{W,F_{\mathrm{stab}}(W)\}.
  \]
\end{lemma}

\begin{proof}
  The size and degree estimates follow immediately from \Cref{lem:stabilizer-atomic-incidence}, just as in the CSS case.
  For a connected component $C$, let $I_C$ be its physical-coordinate set and let $S_C\leq\cP(I_C)$ be the span of its atoms (see \cref{eq:pauli-space}).
  Since the atoms span $S$,
  \[
    S=\bigoplus_C S_C, \qquad L=S^{\perp_s}/S=\mathop{\bigoplus_C} L_C, \qquad L_C:=S_C^{\perp_s}/S_C.
  \]
  Every $L_C$ is nondegenerate.
  Indeed, if a class represented by $u\in S_C^{\perp_s}$ is orthogonal to all of $S_C^{\perp_s}/S_C$, then $u\in(S_C^{\perp_s})^{\perp_s}=S_C$, so the class is zero.

  Let $\PAut_0\leq \PAut(S)$ be the inverse image of the standard subgroup $\rho(\GL(V))$ under the logical action.
  It permutes the nonzero summands $L_C$.
  Let $r$ be the number of nonzero summands.
  The kernel of $\PAut_0\twoheadrightarrow\GL(V)$ acts trivially on $L$ and hence fixes every $L_C$, so the permutation action factors as
  \[
    \PAut_0\twoheadrightarrow\GL(V)\longrightarrow S_r.
  \]
  Every nonzero $L_C$ is nondegenerate and therefore has dimension at least two, giving $2k \ge r\cdot2$ and hence $r\leq k$.
  For $k\geq3$, the group $\GL(V)\cong\GLtwo{k}$ is simple.
  A nontrivial homomorphism $\GLtwo{k}\to S_r$ would therefore be injective, which is impossible because $r\leq k$ and $|\GLtwo{k}|>k!$.
  Hence every $L_C$ is individually $\GL(V)$-invariant.

  By \Cref{lem:stabilizer-invariant-subspaces}, a proper nonzero invariant candidate is $V\oplus0$ or $0\oplus V^*$.
  Both are Lagrangian and have degenerate restricted symplectic form, whereas every $L_C$ is nondegenerate.
  Thus any nonzero $L_C$ must equal all of $L$.
  There is consequently a unique logically active component $\Gamma_*$, and it is $\PAut_0$-invariant.

  Let $H$ be the image of $\PAut_0\to\Aut(\Gamma_*)$.
  Two elements having the same restriction to $\Gamma_*$ induce the same action on $L_*=L$.
  The logical action therefore factors through a surjection $H\twoheadrightarrow\GLtwo{k}$.
\end{proof}

The remainder of the CSS proof is insensitive to the CSS decomposition: it uses only the bounded-degree graph and its large logical quotient.
We therefore obtain the generalized length lowerbound.

\begin{lemma}[Stabilizer length lowerbound]
  \label{lem:stabilizer-n-lowerbound}
  For every integer $W\geq1$, there is a constant $C_W$ such that every binary phantom stabilizer code $S$ encoding $k\geq3$ qubits with $w(S)\leq W$ and $d(S)>W$ satisfies \[ n\geq 2^{k(k-1)/2-C_Wk}. \]
  In particular, $\log n=\Omega_W(k^2)$.
\end{lemma}

\begin{proof}
  By \Cref{lem:active-stabilizer-component}, the active component $\Gamma_*$ has maximum degree at most $\Delta_{\mathrm{stab}}(W)$, at most $(1+F_{\mathrm{stab}}(W))n$ vertices, and an automorphism subgroup $H$ surjecting onto $\GLtwo{k}$.
  Apply \Cref{lem:bounded-degree-quotient,lem:large-prime-part} to any vertex of $\Gamma_*$, with $\Delta=\Delta_{\mathrm{stab}}(W)$.
  This gives
  \[
    (1+F_{\mathrm{stab}}(W))n \geq |V(\Gamma_*)| \geq 2^{k(k-1)/2-C_{\Delta_{\mathrm{stab}}(W)}k},
  \]
  and the factor depending only on $W$ can be absorbed by enlarging the constant in the exponent.
\end{proof}

\stabilizernogo*

\begin{proof}
  The proof of \Cref{thm:css-no-go} applies verbatim after passing to an infinite subsequence on which the bounded integer $w(S_t)$ is constant and using \Cref{lem:stabilizer-n-lowerbound}.
  The hypothesis $k_t\geq3$ holds automatically for all sufficiently large $t$ because there are only finitely many stabilizer codes of bounded length.
\end{proof}

\section{Conclusions and open questions}\label{sec:conclusions}

The results give a sharp scaling-level picture of binary phantom stabilizer codes.
The hardness and constructions are witnessed within the CSS subclass, while the distance obstruction holds for the full stabilizer class.

\begin{center}
  \renewcommand{\arraystretch}{1.25}
  \begin{tabular}{ll}
    \toprule
    Regime for phantom stabilizer codes           & Status established here                                \\
    \midrule
    Recognition, even for $k=2$                   & GI-hard, already on CSS inputs                         \\
    $w,\Delta\leq9$, $k=\Theta(\log n)$, $d=3$    & CSS qLDPC witnesses, $n=7(2^k-1)(2^k-1-k)$             \\
    $w,\Delta=O(1)$, $k=\Theta(\log n)$, $d\ge D$ & CSS qLDPC witnesses at distance at least any fixed $D$ \\
    $w=O(1)$, $k=\omega(\sqrt{\log n})$           & General stabilizers satisfy $d\leq w$ eventually       \\
    \bottomrule
  \end{tabular}
\end{center}

The construction and no-go theorem are complementary statements about phantom stabilizer codes.
For any prescribed constant distance $D$, the degree-split CSS construction followed by fixed-block concatenation yields a rate optimal qLDPC family of phantom stabilizer codes with both check weight and qubit degree bounded in terms of $D$ alone.
Conversely, in every bounded-check-weight family of binary phantom stabilizer codes with $k=\omega(\sqrt{\log n})$, the distance cannot exceed the intrinsic check weight eventually.
Thus qLDPC sparsity and maximal logarithmic logical scaling are compatible with phantomness, but only at bounded distance; more generally, any growing distance above the square-root-logarithmic threshold forces stabilizer-generator weight itself to grow.

There are several natural directions beyond the present work.

\begin{problem}[Exact complexity]
  Determine which complexity class the phantomness recognition problem belongs to, or identify a specific decision problem to which it is polynomial-time equivalent.
  \Cref{thm:GI-hard} proves GI-hardness even for $k=2$ and CSS inputs, but does not settle this classification.
\end{problem}

\begin{problem}[Below the no-go threshold]
  The quantitative theorem allows $d>W$ only when $n$ is at least $2^{\Omega_W(k^2)}$.
  Is this scale attainable?
  More generally, characterize the optimal tradeoff among $n$, $k$, $w$, and $d$ for phantom stabilizer codes when $k=O(\sqrt{\log n})$.
\end{problem}


\begin{problem}[Full phantomness beyond the present model]
  Can either obstruction be avoided while retaining the full standard logical $\GLtwo{k}$ action by enlarging the coding or implementation model?
  In particular, consider the following possibilities.
  \begin{enumerate}
    \item
          Consider more general, possibly non-stabilizer subspace codes or subsystem codes.
    \item
          Enlarge the admissible physical permutations by introducing Floquet-style code switching.
  \end{enumerate}
\end{problem}

\begin{problem}[Partial phantomness]
  Can either obstruction be avoided by requiring only a task-specific subgroup rather than full phantomness?
  Two natural targets are:
  \begin{enumerate}[label=(\roman*)]
    \item
          \emph{SWAP phantomness}, in which physical permutations realize every permutation of a fixed logical basis, as in \Cref{sec:logical-relabelling};
    \item
          \emph{GHZ-preparation phantomness}, in which a physical permutation need only realize the logical transformation \[ |+\rangle|0\rangle^{\otimes(k-1)}\longmapsto\frac{|0^k\rangle+|1^k\rangle}{\sqrt2}. \]
  \end{enumerate}
\end{problem}

\section*{Acknowledgment}

We thank Jonathan Allcock for helpful discussions.
The GI-hardness result for SWAP phantomness (\Cref{sec:logical-relabelling}) and the open problem on partial phantomness are due to him.

The results of this paper were developed primarily through interactions with ChatGPT and Rethlas \cite{ju2026automated}, which provided the main technical approaches.
Initial prompts about the authors' conjecture that phantomness recognition is NP-complete led ChatGPT 5.6-Sol to suggest classical code equivalence; a follow-up request produced the GI-hardness reduction.
Prompts about the authors' conjecture that bounded check weight forces bounded logical dimension produced a counterexample, and further requests for bounded qubit degree led to the qLDPC construction.
Attempts to obtain growing distance yielded only constant-distance constructions, prompting the authors to conjecture the distance no-go theorem.
Rethlas, using Codex with \texttt{gpt-5.6-sol}, produced the CSS proof, and further interactions in ChatGPT extended it to general stabilizer codes.
The initial manuscript draft was also generated by ChatGPT, which the authors manually checked and revised.

With Archon's assistance \cite{ju2026automated}, the authors formalized and machine-checked in Lean the many-one reduction (excluding its polynomial-time assertion), the complete qLDPC construction, and the CSS distance no-go theorem.
The Lean repository will be made public soon.

Throughout the investigation, the authors formulated questions and conjectures, evaluated proposed arguments, and redirected unsuccessful approaches.
The authors take full responsibility for the correctness and presentation of the final manuscript.

\bibliographystyle{halpha}
\bibliography{refs}
\end{document}